\documentclass[10pt]{article}
\pdfoutput=1
\usepackage{bbm}
\usepackage{color}
\usepackage{authblk}
\usepackage{latexsym}
\usepackage{epsfig,amssymb,euscript, mathrsfs,cite}
\usepackage{amsmath}
\usepackage{mathrsfs} 
\usepackage{enumerate}
\definecolor{MyDarkBlue}{rgb}{0.15,0.15,0.45}
\usepackage[linktocpage=true]{hyperref}
\hypersetup{
colorlinks=true,
citecolor=MyDarkBlue,
linkcolor=MyDarkBlue,
urlcolor=MyDarkBlue,
pdfauthor={},
pdftitle={},
pdfsubject={hep-th}
}

\usepackage{hyperref}

\newsavebox{\ns}
\newsavebox{\dbrane}
\newsavebox{\dbshort}

\def\be{\begin{equation}}
\def\ee{\end{equation}}
\def\bea{\begin{eqnarray}}
\def\eea{\end{eqnarray}}

\def\eq#1 { \begin{equation} #1 \end{equation} }

\newlength{\sswidth}

\numberwithin{equation}{section}       

\usepackage{amsmath}
\usepackage{amsthm}
\usepackage{thmtools}
\usepackage{amsfonts}
\usepackage{amssymb}

\newtheorem{lemma}{Lemma}
\newtheorem{theorem}{Theorem}
\newtheorem{corollary}{Corollary}

\newtheorem{definition}{Definition}

\newcommand*\thmref[1]{Theorem \ref{#1}}

\newcommand*\corref[1]{Corollary \ref{#1}}
\newcommand*\lemref[1]{Lemma \ref{#1}}

\newcommand*\BY[1]{\qquad \text{by #1}}

\newcommand*\setdef[3][]{\mathopen #1 \lbrace\, #2 \mathrel{ #1 \vert} #3 \, \mathclose #1 \rbrace}
\newcommand*\set[2][]{\mathopen #1 \lbrace #2 \mathclose #1 \rbrace}
\DeclareMathOperator\re{Re} 

\newcommand*\Arg[2][]{\mathopen #1 ( #2 \mathclose #1 )} 

\newcommand*\qm[1][!]{\ifx#1! \mu \else \mu_{#1} \fi} 
\newcommand*\Prob{P} 

\newcommand*\hi[1][!]{\ifx#1! \gamma \else \gamma_{ #1 } \fi} 
\newcommand*\HS[1][!]{\ifx#1! \Omega \else \Omega_{#1} \fi} 
\newcommand*\Eset[1]{\mathfrak{#1}} 
\newcommand*\EA[1][!]{\ifx#1! \Eset{A} \else \Eset{A}_{#1} \fi} 
\newcommand*\QuasiSystem[1][!]{\ifx#1! (\HS,\EA,f) \else (\HS[#1],\EA[#1],f_{#1}) \fi} 
\newcommand*\System[1][!]{\ifx#1! (\HS,\EA,D) \else (\HS[#1],\EA[#1],D_{#1}) \fi} 

\newcommand*\Sys[1][!]{\ifx#1! \Psi \else \Psi_{#1} \fi} 

\newcommand*\SysSpace[1]{\mathcal{#1}} 
\newcommand*\Strong{\SysSpace{S}} 
\newcommand*\Weak{\SysSpace{W}} 
\newcommand*\Herm{\SysSpace{Q}} 
\newcommand*\PosEntry{\SysSpace{R^+}} 

\newcommand*\dual[1]{\widehat{#1}}

\newcommand*\ev[1]{\mathrm{S}^{\mathrm{e}}_{#1}} 
\newcommand*\od[1]{\mathrm{S}^{\mathrm{o}}_{#1}} 

\newcommand*\evensum{\sigma^{\mathrm{ee}}} 
\newcommand*\eo{\sigma^{\mathrm{eo}}} 

\begin{document}

\begin{titlepage}

\vfill



\begin{center}
   \baselineskip=16pt
   {\Large\bf An argument for strong positivity of the decoherence functional in the path integral approach to the foundations of quantum theory}
  \vskip 1.5cm
Fay Dowker${}^{a,b}$ and Henry Wilkes${}^a$\\
     \vskip .6cm
            \begin{small}
      \textit{
      		${}^{a}${Blackett Laboratory, Imperial College, Prince Consort Road, London, SW7 2AZ, UK}\\ \vspace{5pt}
${}^{b}${Perimeter Institute, 31 Caroline Street North, Waterloo ON, N2L 2Y5, Canada}
              }
              \end{small}                       \end{center}
\vskip 2cm
\begin{center}
\textbf{Abstract}
\end{center}
\begin{quote}
We give an argument for strong positivity of the decoherence functional as the correct, physical positivity condition in formulations of quantum theory based fundamentally on the path integral. We extend to infinite systems work by  Boes and Navascues that shows that the set of strongly positive quantum systems is maximal amongst sets of systems that are closed under tensor product composition. We show further that the set of strongly positive quantum systems is the unique set that is maximal amongst sets that are closed under tensor product composition.
\end{quote}

\vfill

\end{titlepage}

\tableofcontents


\section{Introduction}\label{sec:intro}

The huge breadth of Roger Penrose's work means that there are many topics that are appropriate to include in a volume celebrating his work and his 2020 Nobel Prize for Physics.  His accomplishments range from his pioneering work on global causal analysis in General Relativity that the Nobel Prize recognises, to twistor theory,  quantum foundations and other highly original work in mathematics as well as physics. This paper describes work in quantum foundations, though we believe it is also a contribution to the quest to find a theory of quantum gravity, one of Roger's longstanding interests from a time well before it became mainstream.  This paper adds a technical result to our knowledge about the foundations of the path integral approach to quantum theory, one of whose aims is to answer the question: what is the physical quantum world? Roger's approach to answering this question led to his proposal for a theory in which the wave function or state vector for a quantum system undergoes a dynamical process of collapse induced by the system's interaction with the gravitational field \cite{Penrose:1996cv}.  The path integral offers an alternative perspective in which the physical world is not a state vector or wave function at all, dynamically collapsing or not, and, though Roger has not to our knowledge entertained a path integral approach, we believe that it accords with other aspects of his seminal work --- in particular on the Lorentzian and causal structure of spacetime --- because in the path integral approach the concepts of event and (real time) history are primary,  as they are in General Relativity. 

The path integral can be thought of as the basis of an approach to quantum foundations that takes heed of relativity's lessons. Indeed, the ``fork in the road'' between a Hamiltonian based, canonical approach and a Lagrangian based, path-integral, relativity-friendly approach 
to quantum theory was recognised by Paul Dirac in the early days of quantum mechanics in the paper
\textit{The Lagrangian in Quantum Mechanics} \cite{Dirac:1933}. 
Richard Feynman developed the path integral further  \cite{Feynman:1948} and promoted a way of thinking about quantum theory in which events and histories are central \cite{Feynman:1965, FeynmanQED}. 
In more recent times,  the path integral approach to quantum foundations has been taken up as part of the quest for a solution to the problem of quantum gravity, exactly because the central concepts of event and history in the path integral approach align with those of relativity and because the approach naturally accommodates events involving topology change such as the creation of the universe from nothing and the pair production of black holes. Moreover, in the specific case of  the causal set programme for quantum gravity\footnote{Roger's work on global causal analysis is part of the foundations of causal set theory because it tells us how much information is encoded in the spacetime causal order: in particular, the causal order determines the chronological order \cite{Kronheimer:1967} and the chronological order determines the topology in strongly causal spacetimes \cite{Penrosebook:1972}.}  the characteristic kind of spatio-temporal discreteness of a causal set ``militates
strongly against any dynamics resting on the idea of Hamiltonian evolution'' \cite{Dowker:2010ng} and 
practically demands a histories-based treatment.

The two most developed path integral approaches to quantum foundations are the closely related programmes of  generalised quantum mechanics (GQM) proposed and championed by Jim Hartle \cite{Hartle:1992as, Hartle:1998yg, Hartle:2006nx} 
and quantum measure theory  (QMT) proposed and championed by Rafael Sorkin \cite{Sorkin:1994dt, Sorkin:1995nj, Sorkin:2006wq, Sorkin:2007uc, Sorkin:2010kg}.  In this work we address a question about the axioms of path integral based approaches:  what positivity condition should the decoherence functional --- also called the double path integral in the quantum measure theory (QMT)  literature ---  satisfy? The result we prove is applicable both to QMT and GQM because the decoherence functional is a fundamental entity in both. We will discuss how QMT and GQM diverge from each other  in Section \ref{discussion}. 

Within existing unitary quantum theories,  the decoherence functional is, essentially,  the Gram matrix of inner products of a set of appropriate vectors in an appropropriate vector space and therefore such decoherence functionals satisfy a positivity condition known in the literature as strong positivity  which is, essentially, the condition that the decoherence functional is a positive matrix.  
Conversely,  if one starts --- as one does in QMT and GQM --- with a decoherence functional on an algebra of events as the axiomatic basis for the physics of a quantum system, the condition of strong positivity, if adopted as one of the axioms, allows a Hilbert space to be  be constructed, using  that decoherence functional to define an inner product on the free vector space on the event algebra,  which can then be completed \cite{Martin:2004xi}. This  derived \textit{history Hilbert space} can be shown to equal  the standard Hilbert space in non-relativistic quantum mechanics and in finite unitary systems in a physically meaningful sense  \cite{Dowker:2010ng}  and it is conjectured that this is the case in other unitary quantum theories such as quantum field theory. 
The history Hilbert space   
has  been used to imply the Tsirel'son bound for scenarios of experimental 
probabilities in quantum measure theory \cite{Craig:2006ny} and, more generally, to locate scenarios that admit strongly positive decoherence functionals within the NPA hierarchy of semi definite programs \cite{Navascues:2008, Dowker:2013tva}.  The history Hilbert space 
also provides a complex \textit{vector measure} on events, providing an additional toolkit for exploring the question of the
extension of the quantum measure \cite{Sorkin:2012xx, Dowker:2010qh, Surya:2011}. 

These might be reasons enough to adopt strong positivity as the appropriate positivity axiom for the decoherence functional. We certainly want to be able to recover the standard Hilbert space machinery in familiar cases like quantum mechanics. But it's not a fully conclusive argument because we do not know if a Hilbert space is a necessary structure in a path-integral based theory of quantum gravity. In \cite{Boes:2016ihr} Boes and Navascues give an argument for strong positivity based on composability of finite, noninteracting, uncorrelated systems. They show that the class of finite strongly positive systems satisfies a well defined maximality condition: no other system can be added to the set without 
the set losing the property of being closed under tensor product composition. In this work we will extend their result to infinite systems and further show that the set of strongly positive systems is the unique set of quantum 
systems that is maximal and closed under composition.  

\section{Quantum measure theory: a histories-based framework}  \label{HistoriesIntro}

We will work within the  formalism and use the terminology of quantum measure theory. 
Our results will,
however, apply to generalised quantum mechanics (GQM) because they are technical results about decoherence functionals which are also fundamental in GQM. We will review the basic concepts of QMT below and refer readers to  \cite{Sorkin:1994dt, Sorkin:1995nj, Sorkin:2006wq, Sorkin:2007uc, Sorkin:2010kg} for more details. 

\subsection{Event Algebra}

In quantum measure theory, the kinematics of a physical, quantum system is given by 
the set $\Omega$ of \emph{histories} over which the path integral is done.
Each history $\gamma$ in $\Omega$ is as complete a
description of physical reality as is
conceivable in the theory. 
For example, in $n$-particle quantum mechanics, a history is a set of $n$
trajectories in spacetime and in a scalar field theory, a history is a real or complex
function on spacetime. This is not to say that even in these relatively well-known
cases $\Omega$ is easy to determine: work must be done to 
determine for example if the trajectories/fields are continuous or 
discontinuous and by what measure, \textit{etc}. Nevertheless, the concept of the path integral is familiar enough for us to take $\Omega$ as the underlying context for the concept of a quantum system. 

Any physical proposition about
the system corresponds to a subset of $\Omega$ in the obvious way. For example,
in the case of the non-relativistic particle, if $R$ is a region of
space and $\Delta T$ a time interval, the proposition ``the particle is in $R$ during
$\Delta T$'' corresponds to the set of all trajectories that pass through
$R$ during  $\Delta T$.  We adopt the standard terminology of stochastic
processes in which such subsets of $\Omega$ are called \emph{events}. 

An \emph{event algebra} on a sample space $\Omega$ is a non-empty
collection, $\EA$, of subsets of $\Omega$ such that
\begin{enumerate}
\item  $\Omega\setminus \alpha \in
\EA$ for all $\alpha \in \EA$ (closure under complementation), 
\item $\alpha \cup \beta \in
\EA$ for all  $\alpha, \beta \in \EA $ (closure under finite union). 
\end{enumerate}
It follows from the definition that $\emptyset \in \EA$, $\Omega \in \EA$
and
$\EA$ is closed under finite intersections. An event algebra is an \textit{algebra of sets} by a standard definition,
 and a \textit{Boolean algebra}. For events qua propositions about the system, the set operations correspond  to logical combinations of propositions in the usual way: 
union is  ``inclusive or'', intersection is ``and'', complementation is ``not'' \textit{etc.} 

An event algebra $\EA$ is also an algebra in the sense of a vector space over a set of scalars, 
$\mathbb{Z}_2$, 
with intersection as multiplication and symmetric difference as addition:
\begin{itemize}
\item[] $\alpha \cdot \beta := \alpha \cap \beta$, for all $\alpha,  \beta \in \EA$;
\item[] $\alpha + \beta := (\alpha \setminus \beta) \cup (\beta \setminus
\alpha)$, for all $\alpha, \beta \in \EA$.
\end{itemize}
In this algebra, the unit element, $1 \in \EA$, is $1: = \Omega$ and the
zero element, $0\in \EA$, is $0:= \emptyset$. This ``arithmetric'' way of
expressing set algebraic formulae is very convenient and 
we have, for example, that $1+A$ is the complement of $A$ in $\Omega$. 

If an event algebra $\EA$ is also closed under countable unions then
$\EA$ is a $\sigma${\emph{-algebra}} but we will not assume this extra condition 
on the event algebra. 

\subsection{Decoherence functional and quantal measure} \label{DecoherenceFunctional}

A  \emph{decoherence functional} on an event algebra $\EA$ is a map $D
: \EA \times \EA \to \mathbb{C}$ such that,
\begin{enumerate}
\item $D(\alpha,\beta) =
D(\beta,\alpha)^*$ for all $\alpha, \beta \in \EA$ (\emph{Hermiticity});
\item $D(\alpha, \beta \cup \gamma) = D(\alpha,
\beta) + D(\alpha,\gamma)$ for all $\alpha, \beta, \gamma \in \EA$ with $\beta \cap \gamma
= \varnothing$  (\emph{Additivity});
\item $D(\Omega, \Omega)=1$ (\emph{Normalisation});
\item 
$D(\alpha,\alpha) \geq 0$ for all $\alpha \in \EA$ (\emph{Weak Positivity}).
\end{enumerate}

\noindent A \emph{quantal measure} on an event algebra $\EA$ is a map $\mu: \EA
\to \mathbb{R}$ such that,
\begin{enumerate}
\item  $\mu(\alpha) \geq 0$ for all $\alpha \in \EA$
(\emph{Positivity});
 \item $ \mu(\alpha \cup \beta \cup \gamma) - \mu(\alpha \cup \beta) -
\mu(\beta \cup \gamma) - \mu(\alpha \cup \gamma) +\mu(\alpha) +
\mu(\beta) + \mu(\gamma) = 0$, 
for all pairwise disjoint $\alpha, \beta, \gamma \in \EA$ ({\emph{Quantal Sum Rule}});
\item $\mu(\Omega)=1$ (\emph{Normalisation}).
\end{enumerate}

If $D : \EA \times \EA \to \mathbb{C}$ is a decoherence functional
then the map $\mu : \EA \to \mathbb{R}$ defined by
$\mu(\alpha):=D(\alpha,\alpha)$ is a quantal measure. And, conversely, if $\mu$ is a quantal measure
on $\EA$  then there exists (a non-unique) decoherence functional $D$ such that 
$\mu(\alpha)=D(\alpha,\alpha)$ \cite{Sorkin:1994dt}. The relationship between the quantal measure and 
the decoherence functional and their physical significance --- including the question of which is the more primitive concept --- remain to be fully worked out. In this paper we focus on the decoherence functional. A triple, $(\Omega, \EA, D)$, of sample space, event algebra and
decoherence functional will be called  a {\emph{quantum measure system}}, or 
just  {\emph{system}}  for short in what follows. 

We will also need the more general concept of a \textit{quasi-system} which 
we define to be a triple $(\Omega, \EA, f)$, of sample space, event algebra and, what we will call, a 
functional $f: \EA \times \EA \to \mathbb{C}$ that satisfies the conditions (1)---(3) in the above definition of 
a decoherence functional but is not necessarily weakly positive.

We call the set of quasi-systems $\Herm$ and the set of systems $\Weak$.

A comment is in order here about why weak positivity of the decoherence functional is a requirement for a
 physical system. Weak positivity is equivalent to the requirement that the quantum measure, $\mu(\alpha):=D(\alpha,\alpha)$
 take only real, non-negative values.  In the development of our 
 understanding of the quantum measure, the predictive ``law of preclusion'' \cite{Geroch:1984, Sorkin:2007uc}  that 
  events of zero, or of very small measure, are precluded from happening plays an important role. This preclusion law only  makes  sense if the measure is non-negative, since, otherwise, certain events would have lower measure than the zero-measure events. The positivity axiom for generalised quantum mechanics (GQM) is weak positivity of the decoherence functional (see for example condition (ii) on page 32 of \cite{Hartle:2006nx} and equation (2.25a) of \cite{PhysRevD.46.1580}.

\section{Composition} \label{section_composition}

We want to describe a system that is composed of two non-interacting, uncorrelated subsystems. For reasons that will become clear, we define composition at the level of 
quasi-systems. 
Consider two quasi-systems $\QuasiSystem[1]$ and $\QuasiSystem[2]$ that together form a composite quasi-system  $(\Omega, \EA, f)$, which we write
\begin{align*}
(\Omega, \EA, f) = \QuasiSystem[1] \odot \QuasiSystem[2]  = (\HS[1]\odot \HS[2], \EA_1 \odot \EA_2, f_1 \odot f_2)\,,
\end{align*}
where the individual components of the composite triple, $ \HS[1]\odot \HS[2]$, $\EA_1 \odot \EA_2$  and $ f_1 \odot f_2$ are defined below. 

First, we take the composite history space to be the Cartesian product: $\HS[1] \odot  \HS[2]: = \HS[1] \times \HS[2]$. To construct the composite event algebra $\EA$, first consider product events of the form ``$E_1 \in \EA[1]$ for quasi-system $1$ and $E_2 \in \EA[2]$ for quasi-system $2$'',  given by the Cartesian product $E_1 \times E_2$. These product events must be in the composite event algebra, and we define 
$\EA= \EA_1 \odot \EA_2$, to be the event algebra generated by the set of product events \textit{i.e.} $\EA$  is the smallest event algebra that contains all the product events. One can show that
$\EA$ equals the set of finite disjoint unions\footnote{A note of clarification: by a ``disjoint union'', here and throughout the paper, we mean a union of a collection of sets that are pairwise disjoint. The symbol $\sqcup$ denotes disjoint union, \textit{i.e.} it implies that the 
sets whose union is being taken are pairwise disjoint.} of product events. This is standard but we will go through it.

Let us define $\tilde{\EA}$ to be the set of finite disjoint unions of product events. Then $\tilde{\EA}\subseteq \EA$.
All we need to show therefore is that $\tilde{\EA}$ is an algebra. 

\begin{lemma} (Closure under union)
$X \cup Y \in \tilde{\EA}$  for all $X,Y \in \tilde{\EA}$. 
\end{lemma}
\begin{proof}
Let $X$ and $Y$ be elements of $\tilde{\EA}$. They are finite disjoint unions of product events and so 
their union is a finite union of product events. Thus,  if we show that any event $Z \in \EA$ of the form 
\begin{equation}
		Z = \bigcup_{i=1}^{n} {V}_i \times {W}_i\,,
		\end{equation}
where $V_i \in \EA_1$ and $W_i \in \EA_2$ for all $i$, equals a finite union of 
pairwise disjoint product events then we are done. 

Consider the algebra $\EA_V\subseteq \EA_1$ generated by $\{ {V}_1, \dots {V}_n \} $. It is a finite Boolean algebra. Let the atoms of this algebra be 
$\{v_1, \dots v_N\}$. Similarly consider the algebra $\EA_W \in \EA_2$ generated by $\{ {W}_1, \dots {W}_n \} $. Let the atoms of this algebra be 
$\{w_1, \dots w_M\}$. The atoms of the product algebra $\EA_V \otimes \EA_W$ are of the form 
$v_k \times w_l$. The event $Z$ is an element of $\EA_V \otimes \EA_W$ and so it is a unique disjoint union of
atoms of the form $v_k \times w_l$. 
\end{proof}

\begin{lemma} (Closure under complementation)
$1+ X$ is an element of $\tilde{\EA}$  for all $X \in \tilde{\EA}$. 
\end{lemma}

\begin{proof}
If $X$ be an element of $ \tilde{\EA}$ then it is a disjoint union of product events:
\begin{equation}
X = \bigsqcup_{i=1}^{n} {V}_i \times {W}_i \,.
\end{equation}
So $X$ is an element of the product algebra $\EA_V \otimes \EA_W$ as constructed in the proof of the previous lemma. 
$1+ X $ is also an element of $\EA_V \otimes \EA_W$ and so it is a disjoint union of the product atoms of $\EA_V \otimes \EA_W$. 
\end{proof}

\begin{corollary} $\EA = \tilde\EA$. 
\end{corollary}

If one thinks of the event algebras as algebras qua vector spaces over   $\mathbb{Z}_2$
 then one sees that  $\EA$ is the tensor product $ \EA_1 \otimes \EA_2$.
 
Finally, we define the composed functional $f$ following \cite{Martin:2004xi} and \cite{Boes:2016ihr}. We assume that the two subsystems do not interact and are uncorrelated. In a classical measure theory the probability of the product event of two independent events is $\Prob_1(E_1) \Prob_2(E_2)$, where $\Prob_1$ and $\Prob_2$ are the probability measures for system $1$ and $2$, respectively. By analogy, we define $f = f_1\odot f_2$ for product events:
\begin{equation}
	f(A_1 \times A_2, B_1 \times B_2) := f_1(A_1,B_1)\, f_2(A_2,B_2) \,. \label{comp_rule_basic}
\end{equation}
One might want to consider other ways to combine $f_1$ and $f_2$ but note that 
if the probability measures $\Prob_1$ and $\Prob_2$ were expressed as two diagonal decoherence functionals $D_1$ and $D_2$, then this composition rule reproduces the classical composition rule. Moreover, such a composition rule is observed for decoherence functionals constructed in ordinary quantum mechanics when the initial state of the combined system is a product state.

The functional $f$ is extended to the rest of $\EA = \EA_1 \otimes \EA_2$ by linearity. Consider two arbitrary elements
of $\EA$, 
\begin{equation}
		A = \bigsqcup_{i=1}^{n_A} {A_1}_i \times {A_2}_i \quad \textrm{and} \quad B = \bigsqcup_{j=1}^{n_B} {B_1}_j \times {B_2}_j \,, 
\end{equation}
where the notation $\sqcup$ indicates that the sets over which the union is taken are pairwise disjoint. 

We extend $f$ to these events: 
\begin{equation}\label{eq:dab}
	f(A,B) := \sum_{i=1}^{n_A} \sum_{j=1}^{n_B} f_1({A_1}_i,{B_1}_j)\, f_2({A_2}_i, {B_2}_j)\,, 
\end{equation}
where we must check that this 
is independent of the expansions of $A$ and $B$ as disjoint unions of products. 

Consider therefore different expansions of $A$ and $B$ as disjoint unions:
\begin{equation}
		A = \bigsqcup_{k=1}^{m_A} {\alpha_1}_k \times {\alpha_2}_k \quad \textrm{and} \quad B = \bigsqcup_{l=1}^{m_B} {\beta_1}_l \times {\beta_2}_l \,.
\end{equation}
Let $\EA_{1A}$ be the event algebra generated by the events $\{{A_1}_1, {A_1}_2, \dots {A_1}_{n_A}\} \cup 
\{{\alpha_1}_1, {\alpha_1}_2, \dots {\alpha_1}_{m_A}\}$. Let $\EA_{2A}$ be the event algebra generated by the events $\{{A_2}_1, {A_2}_2, \dots {A_2}_{n_A}\} \cup 
\{{\alpha_2}_1, {\alpha_2}_2, \dots {\alpha_2}_{m_A}\}$.

Now, $A$ is an element of $\EA_{1A} \otimes \EA_{2A}$  and has a unique 
expansion as a disjoint union of atoms of this algebra. Each of these atoms is a product of an atom of $\EA_{1A}$
and an atom of $\EA_{2A}$. 

We can go through a similar procedure for $B$, defining algebras $\EA_{1B}$,  $\EA_{2B}$, and  $\EA_{1B} \otimes \EA_{2B}$ and their atoms. 

Then, starting with $f(A,B)$ as defined by equation (\ref{eq:dab}), 
and using the additivity of $f_1$ and $f_2$ separately, 
we can re-express this as a unique double sum over the atoms of $\EA_{1A} \otimes \EA_{2A}$ and over the atoms of 
$\EA_{1B} \otimes \EA_{2B}$. Then, again using the additivity of $f_1$ and $f_2$, 
those atoms can be recombined to form the events ${\alpha_1}_j \times {\alpha_2}_j$ and 
$ {\beta_1}_j \times {\beta_2}_j $ to show that 
\begin{equation}\label{eq:dabb}
  \sum_{i=1}^{n_A} \sum_{j=1}^{n_B} f_1({A_1}_i,{B_1}_j)\, f_2({A_2}_i, {B_2}_j) = 
   \sum_{k=1}^{m_A} \sum_{l=1}^{m_B} f_1({\alpha_1}_k,{\beta_1}_l)\, f_2({\alpha_2}_k, {\beta_2}_l) 
  \,,
\end{equation}
so  $f$ is well-defined. This completes our definition of the composition of quasi-systems.

We have defined composition for quasi-systems because, it turns out, composition does not 
preserve weak positivity: the composition of two systems may not be a system. 
For example, consider two finite systems each with only two histories: 
$\Omega_1 = \{\gamma^{(1)}_1, \gamma^{(2)}_1\} $ and 
$\Omega_2 = \{\gamma^{(1)}_2, \gamma^{(1)}_2\} $.
For each system, the atomic events are the singleton sets with one element. Consider, for each system the set of atomic events and let the respective $2\times 2$ matrices  $M$ and $N$ be
\begin{align*}
	M_{ij} : = &D_1(\set{\gamma^{(i)}_1}, \set{\gamma^{(j)}_1}) \, ,\\
         N_{ij} : = &D_2(\set{\gamma^{(i)}_2}, \set{\gamma^{(j)}_2 }) \, , 
	\end{align*}
and have entries
\begin{equation}\label{weakexample}
	M = \begin{pmatrix}
		2 & -1
		\\ -1 & 1
	\end{pmatrix}
	\quad \textrm{and} \quad N= \frac{1}{5} \begin{pmatrix}
		1 & 2
		\\ 2 & 0
	\end{pmatrix} \,.
\end{equation}

Consider now the composed event  $E : = E_{1} + E_{2}$ where
 $E_{1}:=\set{(\gamma^{(1)}_1, \gamma^{(1)}_2)}$ and 
$E_{2}:=\set{(\gamma^{(2)}_1, \gamma^{(2)}_2)}$. We have the composed functional
\begin{align*}
& D_1\odot D_2 (E, E)  \\
& =  D_1\odot D_2(E_{1},E_{1}) + D_1\odot D_2(E_{1},E_{2}) + D_1\odot D_2(E_{2},E_{1}) + D_1\odot D_2(E_{2},E_{2})
	\\ 
	 &= {M}_{11} {N}_{11} + {M}_{12}{N}_{12} + {M}_{21}{N}_{21} + {M}_{22}{N}_{22}
	\\ &= \frac{1}{5} (2 - 2 -2 + 0)
	\\ &= -\frac{2}{5}\,.
\end{align*}
$D_1$ and $D_2$ are weakly positive but $D_1\odot D_2$ is not and so the set of quantum measure systems, $\Weak$, is not closed under composition. Therefore, if we require that any two physical systems must compose to form a physical system then the conclusion is that not all systems in $\Weak$ are physical. We can turn this around and impose ``membership of a class 
of systems that is closed under composition'' as a requirement to be an allowed physical system. 

\begin{definition}[Tensor-Closed]
	A subset $\SysSpace{A} \subseteq \Weak$ is tensor-closed if
	\begin{equation}
		\Sys[1] \odot \Sys[2] \in \SysSpace{A} \quad \forall \Sys[1],\Sys[2] \in \SysSpace{A}\,.
	\end{equation}
\end{definition}
We have chosen to call this property tensor-closed because the composed event algebra is the 
tensor product algebra. 

The question to investigate is then, what subsets of $\Weak$ are tensor-closed? One such subset 
has already been identified in the literature: the set of  systems with \textit{strongly positive} decoherence
functionals \cite{Martin:2004},  to which we now turn.

\section{Strong Positivity} \label{section_strong_pos}

\begin{definition}[Event Matrix]
	Given a functional $f: \EA \times \EA \to {\mathbb{C}}$ and a finite set of events $\Eset{B} \subseteq \EA$, the corresponding Hermitian event matrix $M$ is the $\vert \Eset{B} \vert \times \vert \Eset{B} \vert $ 
	square matrix, indexed by $\Eset{B}$, given by
	\begin{equation}
		M_{AB}: = f(A,B)\,, \quad A, B \in \Eset{B} \,.
	\end{equation}
\end{definition}

Using this concept of event matrix, the definition of strong positivity can be stated:
\begin{definition}[Strong Positivity]
	A decoherence functional $D:  \EA \times \EA \to {\mathbb{C}}$ is strongly positive if, for each finite $\Eset{B} \subseteq \EA$, the corresponding event matrix $M$ is positive semi-definite. 
	\end{definition}

 This condition is strictly stronger than weak positivity, indeed event matrix $M$ in (\ref{weakexample}) above is weakly positive but not positive definite. We call a system with a 
 strongly positive decoherence functional a strongly positive system and 
 denoting the set of all strongly positive systems $\Strong$, we have 
\begin{align*}
\Strong \subset \Weak \subset \Herm\,.
\end{align*} 

We will prove that  $\Strong$ is tensor-closed using the following lemma:
\begin{lemma} \label{sum_of_strong_pos_events_is_strong_pos}
	Consider a system $\System$ and the finite set of events $\Eset{B} \subseteq \EA$ with event matrix $M$. If there exists a finite set of events $\Eset{B}'\subseteq \EA$ such that  the event matrix $M'$ of $\Eset{B}'$ is positive semi-definite and 
 every event in $\Eset{B}$ is a finite disjoint union of events in $\Eset{B}'$, then the event matrix $M$ of $\Eset{B}$ is also positive semi-definite.
\end{lemma}
\begin{proof}
	A similar claim can be found on page 8 of \cite{Martin:2004}  and we follow the same method of proof. By assumption, for each event $E \in \Eset{B}$ there is a number $n_{E}$ such that $E$ is a union of $n_E$ pairwise disjoint events $E_i \in \Eset{B}'$:
	\begin{equation}
		E = \bigsqcup_{i=1}^{n_E} E_i \,.
	\end{equation}
	Now, for any $v \in {\mathbb{C}}^{\vert \Eset{B} \vert}$,
	\begin{align*}
		v^\dagger M v &= \sum_{A \in \Eset{B}} \sum_{B \in \Eset{B}} {v_A}^* D(A,B)\, v_B
		\\ &= \sum_{A \in \Eset{B}} \sum_{B \in \Eset{B}} {v_A}^* v_B \sum_{i=1}^{n_A} \sum_{j=1}^{n_B} D(A_i,B_j) \BY{bi-additivity}
		\\ &= \sum_{A \in \Eset{B}} \sum_{B \in \Eset{B}} {v_A}^* v_B \sum_{i=1}^{n_A} \sum_{j=1}^{n_B} M'_{A_i B_j}
		\\ &= \sum_{A \in \Eset{B}} \sum_{B \in \Eset{B}} {v_A}^* v_B \sum_{i=1}^{n_A} \sum_{j=1}^{n_B} \sum_{E \in \Eset{B}'} \sum_{F \in \Eset{B}'} \delta_{A_i E}\, \delta_{B_j F}\, M'_{EF}
		\\ &= \sum_{E \in \Eset{B}'} \sum_{F \in \Eset{B}'} \Arg[\Bigg]{\sum_{A \in \Eset{B}} {v_A}^* \sum_{i=1}^{n_A} \delta_{A_i E}} M'_{EF} \Arg[\Bigg]{\sum_{B \in \Eset{B}} v_B \sum_{j=1}^{n_B} \delta_{B_j F}}
		\\ & = \sum_{E \in \Eset{B}'} \sum_{F \in \Eset{B}'} V_E^* M'_{EF} V_F\,,
	\end{align*}
where  vector $V \in {\mathbb{C}}^{\vert \Eset{B}'\vert}$ and its components are $V_F := \Arg[\Bigg]{\sum_{B \in \Eset{B}} v_B \sum_{j=1}^{n_B} \delta_{B_j F}}$.  $M'$ is positive semi-definite and so
$V^\dagger M' V \ge 0$. Hence the result. 
\end{proof}

\begin{corollary}Strong positivity is preserved under finite coarse-graining.\end{corollary}

\begin{theorem} \label{strong_pos_composition}
	If $\Sys[1]=\System[1]$ and $\Sys[2]=\System[2]$ are strongly positive systems, then  $\Sys[1] \odot \Sys[2]$ is a strongly positive system.
\end{theorem}
\begin{proof}
	Consider a  set of events $\Eset{B} \subseteq \EA[1] \odot \EA[2]$ of cardinality $n$:
	$\Eset{B} = \{ X^{1}, X^2, \dots X^n\}$. 
	By the previous Lemma, if there exists a set of events $\Eset{B}' \subseteq \EA[1] \odot \EA[2]$
	with a positive definite event matrix, such that every element of $\Eset{B}$ is a disjoint union 
	of elements of $\Eset{B}'$ then we are done. 
	
	Each element of $\Eset{B}$ is a disjoint union of product events:
		\begin{align*}
		X^a =  \bigsqcup_{i=1}^{n_a} X^a_{1i }\times X^a_{2i} \,, \quad a = 1,2, \dots n\,.
	\end{align*}
Let  $\EA(\Eset{B})_1$ be the subalgebra of $\EA[1]$  generated by the set of events 
	$\{ X^a_{1i }\, | \, a = 1,2, \dots n \ \textrm{and} \ i = 1,2, \dots n_a \}$. 
Similarly let $\EA(\Eset{B})_2$ be the subalgebra of $\EA[2]$  generated by the set of events 
	$\{ X^a_{2i }\, | \, a = 1,2, \dots n \ \textrm{and} \ i = 1,2, \dots n_a \}$. 
	Consider the product algebra, $\EA(\Eset{B})_1 \otimes \EA(\Eset{B})_2$. 
	Its atoms are products of the form $a_{1i} \times a_{2j}$ where 
	$a_{1i}$, $i = 1,2,\dots m_1$ are the atoms of  $\EA(\Eset{B})_1$ 
	and $a_{2j}$, $j = 1,2, \dots m_2$ are the atoms of $\EA(\Eset{B})_2$. 
	 Let $\Eset{B}' $ denote the set of these product atoms: 
	\begin{align*}
		\Eset{B}' =  \{ a_{1i} \times a_{2j}\, |\, i = 1,2\dots m_1\ \textrm{and} \ j = 1,2,\dots m_2 \}\,.
			\end{align*} 
	 	Each event $X^a \in \Eset{B}$ is a unique finite disjoint union of elements of $\Eset{B}' $.
	The event matrix for $\Eset{B}'$ is 
	\begin{align*}
		D_1 \odot D_2 (a_{1i} \times a_{2j} \,,\, a_{1k}  \times a_{2l} ) = D_1 (a_{1i} , a_{1k} ) \, D_2 (
		a_{2j}  , a_{2l} )\,.
	\end{align*}
	This is the Kronecker product of two positive semi-definite 
	matrices, which is positive semi-definite. Hence the result. 
	\end{proof}

Thus, $\Strong$ is tensor-closed. However this condition is not sufficient to pick out $\Strong$ uniquely from amongst subsets of 
$\Weak$. 
\begin{definition}[Positive Entry Decoherence Functional]
	A decoherence functional $D: \EA \times \EA \to {\mathbb{C}}$ is a positive entry decoherence functional if, for all $A,B \in \EA$, $D(A, B)$ is real and non-negative. \end{definition}

We call a system with a positive entry decoherence functional a positive entry system. Let $\PosEntry$ denote the set of positive entry systems: $\PosEntry \subset \Weak$. The composition of two positive entry systems is a  positive entry system:
\begin{lemma} \label{pos_entry_composition}
$\PosEntry$ is tensor-closed.
\end{lemma}
\begin{proof}  
Let  $\Sys[1]=\System[1]$ and $\Sys[2]=\System[2]$ be positive entry systems.

	Each event $E \in \EA[1] \odot \EA[2]$ can be expanded as a finite disjoint union
	\begin{equation}
		E = \bigsqcup_{i=1}^{n_E} {E_1}_i \times {E_2}_i\,,
	\end{equation}
	where ${E_1}_i \in \EA[1]$ and ${E_2}_i \in \EA[2]$. Then
	\begin{align*}
		D_1 \odot D_2 (A,B) &= \sum_{i=1}^{n_A} \sum_{j=1}^{n_B} \underbrace{D_1({A_1}_i,{B_1}_j)}_{\geq 0}\, \underbrace{D_2({A_2}_i,{B_2}_j)}_{\geq 0} \label{pos_entry_comp_geq_0}
		\\ & \geq 0 \,.
	\end{align*}
\end{proof}

Another example of a tensor-closed set of systems is the set of classical systems: a system is classical if there exists a classical (probability) measure $\mu$ on $\EA$ such that 
$D(A,B) = \mu(A \cap B)$ for all $A$ and $B$ in $\EA$. 

\begin{lemma} \label{pos_entry_composition}
The set of classical systems is tensor-closed.
\end{lemma}
\begin{proof}  
Let  $\Sys[1]=\System[1]$ and $\Sys[2]=\System[2]$ be classical systems with corresponding 
classical measures $\mu_1$ and $\mu_2$ respectively.

	Each event $E \in \EA[1] \odot \EA[2]$ can be expanded as a finite disjoint union
	\begin{equation}
		E = \bigsqcup_{i=1}^{n_E} {E_1}_i \times {E_2}_i\,,
	\end{equation}
	where ${E_1}_i \in \EA[1]$ and ${E_2}_i \in \EA[2]$. Then
	\begin{align*}
		D_1 \odot D_2 (A,B) &= \sum_{i=1}^{n_A} \sum_{j=1}^{n_B} D_1({A_1}_i,{B_1}_j) \, D_2({A_2}_i,{B_2}_j)   \label{classsys}
		\\ &=  \sum_{i=1}^{n_A} \sum_{j=1}^{n_B} \mu_1({A_1}_i\cap{B_1}_j) \, \mu_2({A_2}_i\cap {B_2}_j) 	\\
		& = \sum_{i=1}^{n_A} \sum_{j=1}^{n_B}   \mu_1 \otimes \mu_2 (({A_1}_i \times {A_2}_i)\cap ({B_1}_j \times {B_2}_j)) \\
		& =   \mu_1 \otimes \mu_2 (A\cap B)  \,,
	\end{align*}
	where $\mu_1 \otimes \mu_2$ is the product classical measure on the product algebra $\EA_1
	\otimes \EA_2$. 
\end{proof}

\subsection{Galois Self-Dual Sets}

In \cite{Boes:2016ihr}, Boes and Navascues showed that, 
in the case where the set of systems considered is the set of \textit{finite} systems,  $\Weak_{fin}$,   the set of finite strongly positive systems, $\Strong_{fin}$,  is a maximal tensor-closed set:  the set $\Strong_{fin}$ cannot be enlarged to include any system in $\Weak_{fin} \setminus \Strong_{fin}$ and remain tensor-closed.
We will reproduce this result,  extending it to infinite systems $\Weak$ and
$\Strong$. We will formalise the maximality condition using the concept of  \textit{Galois dual}\footnote{We thank Rafael Sorkin for introducing us to this concept}:
\begin{definition}[Galois Dual]
	The Galois dual of a subset $\SysSpace{A} \subseteq \Weak$ is the set
	\begin{equation}
		\dual{\SysSpace{A}} := \setdef{ \Sys[1] \in \Weak} { \Sys[1] \odot \Sys[2] \in \Weak \quad \forall \Sys[2] \in \SysSpace{A}} \,.
	\end{equation}
\end{definition}
\begin{definition}[Galois Self-Dual]
	A subset $\SysSpace{A} \subset \Weak$ is Galois self-dual if $\dual{\SysSpace{A}} = \SysSpace{A}$.
\end{definition}
In other words, the Galois dual of  a set of systems $\SysSpace{A}$  is the set of 
systems whose composition with \textit{any} element of $\SysSpace{A}$ is also a system.

Note: the term ``Galois'' dual refers to the fact that the Galois dual operation, together with itself, is an \textit{antitone Galois connection}. Indeed, $\SysSpace{A} \subseteq \dual{\SysSpace{B}} \Leftrightarrow \SysSpace{B} \subseteq \dual{\SysSpace{A}}$.

\begin{lemma} \label{composition_closed_subset_of_dual}
	If $\SysSpace{A} \subseteq \Weak$ is tensor-closed, then ${\SysSpace{A}} \subseteq \dual{\SysSpace{A}}$.
\end{lemma}
\begin{proof}
	Consider $\Sys[1] \in \SysSpace{A}$. Since $\SysSpace{A}$ is tensor-closed, for all $\Sys[2] \in \SysSpace{A}$,
	\begin{equation}
		\Sys[1] \odot \Sys[2] \in \SysSpace{A} \,.
	\end{equation}
	Therefore, $\Sys[1] \in \dual{\SysSpace{A}}$.
\end{proof}
Amongst the tensor-closed subsets of $\Weak$, a subset $\SysSpace{A}$ that is also Galois self-dual is maximal because there is no system outside of $\SysSpace{A}$ that can be composed with all systems in $\SysSpace{A}$ to produce a system.

\begin{theorem} \label{S_self_dual}
	$\dual{\Strong} = \Strong $.
\end{theorem}
\begin{proof}
	Since $\Strong$ is tensor-closed, by \lemref{composition_closed_subset_of_dual} we have $ \Strong \subseteq \dual{\Strong}$.

	Now consider $\Sys[1]=\System[1] \in \dual{\Strong}$. To prove that $\Sys[1] \in \Strong$ we need to show that, for any finite subset $\Eset{B} \subseteq \EA[1]$, the corresponding event matrix $M_1$ is positive semi-definite.
	
Let $v$ be a vector in $ {\mathbb{C}}^{\vert \Eset{B} \vert}$. We define a square matrix $M_2$, of order 
$(\vert \Eset{B} \vert + 1)$,  indexed by $\Eset{B}' = \Eset{B} \cup \{ x\}$ where $x$ is some extra index value, 
	\begin{equation}
		{M_2}_{AB} := \frac{1}{r} \begin{cases}
			{v_A}^* v_B & \quad \text{if } A,B \in \Eset{B}
			\\ 1 & \quad \text{if } A=B=x
			\\ 0 & \quad \text{otherwise,}
		\end{cases}
	\end{equation}
where
	\begin{equation}
		r = 1 + \sum_{A \in \Eset{B}} {v_A}^* \sum_{B \in \Eset{B}} v_B \,.
	\end{equation}
	Note that the extra index value $x$ and the ${M_2}_{xx}= 1$ entry are necessary to ensure that $r$ is a strictly positive number, in the cases where $\sum_{A \in \Eset{B}} {v_A} = 0$. 
	This matrix is normalised --- in the decoherence functional sense that the sum of its entries equals 1 --- and Hermitian. Moreover, it is positive semi-definite because, for any $u \in {\mathbb{C}}^{\vert \Eset{B} \vert + 1}$,
	\begin{align*}
		u^\dagger {M_2} u &= \frac{1}{r}\,  {u_{x}}^* u_{x} + \frac{1}{r} \sum_{A \in \Eset{B}} \sum_{B \in \Eset{B}} {u_{A}}^* {v_{A}}^* v_B u_B
		\\ &= \frac{1}{r} \vert u_{x} \vert^2 + \frac{1}{r} \Bigg\vert \sum_{A \in \Eset{B}} v_A u_A \Bigg\vert^2
		\\ & \geq 0\,.
	\end{align*}

	Now, consider the system $\Sys[2]=\System[2]$ where $\HS[2] = \{\gamma_\alpha\,|\, \alpha \in \Eset{B}'\}$ is a finite history space indexed by $\Eset{B}'$. The singleton sets 
	$\{\gamma_\alpha\}$, $ \alpha \in \Eset{B}'$, are the atoms of the algebra $\EA_2$. 
	Since $\HS[2]$ is finite, we can define $D_2$ by 
	choosing  $M_2$ as the event matrix for the set of atoms and  $D_2$ is defined by additivity for all the other events. $D_2$ is strongly positive, so $\Sys[2] \in \Strong$. Therefore, since $\Sys[1] \in \dual{\Strong}$, it follows by definition that $\Sys[1] \odot \Sys[2] \in \Weak$, which implies that $D_1 \odot D_2$ is weakly positive.
	
	Now, consider the event $E \in \EA[1] \odot \EA[2]$ given by
	\begin{equation}
		E = \bigsqcup_{A \in \Eset{B}} A \times \set{\gamma_A} \,.
	\end{equation}
Since the $\set{\gamma_A}$ are atoms of $\EA_2$, the union is indeed a disjoint union. Since $D_1 \odot D_2$ is weakly positive, it follows that
	\begin{align*}
		0 & \leq D_1 \odot D_2 (E,E)
		\\ &= \sum_{A \in \Eset{B}} \sum_{B \in \Eset{B}} D_1(A,B)\, D_2(\set{\gamma_A},\set{\gamma_B})
		\\ &= \sum_{A \in \Eset{B}} \sum_{B \in \Eset{B}} {M_1}_{AB} {M_2}_{AB}
		\\ &= \frac{1}{r} \sum_{A \in \Eset{B}} \sum_{B \in \Eset{B}} {v_A}^* {M_1}_{AB} v_B
		\\ &= \frac{1}{r} \, v^\dagger M_1 v \,.
	\end{align*}
	Since $r$ is a positive number and $v$ is arbitrary, this implies that $M_1$ is positive semi-definite. $\Eset{B} \subseteq \EA[1]$ was also arbitrary and so $\Sys[1] \in \Strong$ and  $\dual{\Strong} \subseteq \Strong$.\end{proof}
Boes and Navascues' proof of this result for finite systems is very similar: they use a decoherence functional 
in the role of $D_2$ that is constructed explicitly from strings of projectors and an initial state in a Hilbert space. 	
The next two  lemmas  show that the set of positive entry systems $\PosEntry$ is not Galois self-dual.
\begin{lemma} \label{P_dual_equiv_pos_real}
	Let $\Sys=\System \in \Weak$ be a system. 
	\begin{equation}
		\Sys \in \dual{\PosEntry} \quad \Longleftrightarrow \quad \re\Arg[\big]{D(A,B)} \geq 0 \quad \forall A,B \in \EA \,.
	\end{equation}
\end{lemma}
\begin{proof}
	Let $\Sys[1]=\System[1] \in \Weak$ and $\Sys[2]=\System[2] \in \PosEntry$. Note that this implies $D_2$ is a real symmetric functional.
	
	For ``$\Longleftarrow$'' suppose that 
	\begin{equation}
		\re\Arg[\big]{D_1(A,B)} \geq 0 \quad \forall A,B \in \EA[1] \,.
	\end{equation}
	Then, for any event $E \in \EA[1] \odot \EA[2]$, expanded as the disjoint union
	\begin{equation}
		E = \bigsqcup_{i=1}^{n_E} {E_1}_i \times {E_2}_i \,,
	\end{equation}
	the corresponding diagonal entry in $D_1 \odot D_2$ is
	\begin{alignat}{2}
		& D_1 \odot D_2 (E,E) \nonumber
		\\* &= \sum_{i=1}^{n_E} \sum_{j=1}^{n_E} D_1({E_1}_i,{E_1}_j)\, D_2({E_2}_i,{E_2}_j)
		\\ &= \frac{1}{2} \sum_{i=1}^{n_E} \sum_{j=1}^{n_E} \mathopen \big ( D_1({E_1}_i,{E_1}_j)\, D_2({E_2}_i,{E_2}_j) + D_1({E_1}_j,{E_1}_i)\, \makebox[0pt][l]{$ \displaystyle D_2({E_2}_j,{E_2}_i) \mathclose \big ) $}
		\\ &= \frac{1}{2} \sum_{i=1}^{n_E} \sum_{j=1}^{n_E} \Arg[\big]{ D_1({E_1}_i,{E_1}_j) + D_1({E_1}_j,{E_1}_i) }D_2({E_2}_i,{E_2}_j) && \BY{symmetry of $D_2$}
		\\ &= \sum_{i=1}^{n_E} \sum_{j=1}^{n_E} \underbrace{\re\Arg[\big]{ D_1({E_1}_i,{E_1}_j)}}_{\geq 0} \underbrace{D_2({E_2}_i,{E_2}_j)}_{\geq 0} && \BY{Hermiticity of $D_1$}
		\\ &\geq 0 \,.
	\end{alignat}
	Therefore, $\Sys[1] \odot \Sys[2] \in \Weak$ and so $\Sys[1] \in \dual{\PosEntry}$.

	For ``$\Longrightarrow$'' suppose $\Sys[1] \in \dual{\PosEntry}$.
	Let $\Sys[2]$ have exactly two histories, $\{\gamma_a, \gamma_b\}$, and let the event matrix of the atoms, 
	$\{\gamma_a\}$ and $\{\gamma_b\}$,  be
	 $M = \tfrac{1}{2} \Arg[\big]{ \begin{smallmatrix} 0 & 1 \\ 1 & 0 \end{smallmatrix} }$, so  $ \Sys[2] \in \PosEntry$. 
	 
	Let $A,B \in \EA[1]$, and consider the event $E \in \EA[1] \odot \EA[2]$ given by
	\begin{equation}
		E = A \times \set{\gamma_a} \sqcup B \times \set{\gamma_b}\,.
	\end{equation}
	 Then, since $\Sys[1] \odot \Sys[2] \in \Weak$, it follows that
	\begin{align*}
		0 & \leq D_1 \odot D_2 (E,E)
		\\ &= D_1(A,A) M_{aa} + D_1(B,B) M_{bb} + D_1(A,B) M_{ab} + D_1(B,A) M_{ba}
		\\ &= \frac{1}{2} \Arg[\big]{ D_1(A,B) + D_1(B,A) }
		\\ &= \re\Arg[\big]{D_1(A,B)} \,.
	\end{align*}
\end{proof}

\begin{lemma} \label{P_not_self_dual}
	$\dual{\PosEntry} \supset \PosEntry$.
\end{lemma}
\begin{proof}
	By \lemref{P_dual_equiv_pos_real}, $\dual{\PosEntry}$ equals the set of systems  $\System$ such that $\re(D)$ is a positive entry decoherence functional. This will include all the systems in $\PosEntry$, but will also include, for example, the system with two histories whose ``atomic'' event  matrix is $\tfrac{1}{2} \Arg[\big]{\begin{smallmatrix} 1 & i \\ -i & 1 \end{smallmatrix} }$. This is not a positive entry system. 
	\end{proof}

\subsection{Self-Composition}

We will prove that $\Strong$ is the only subset of $\Weak$ that is tensor-closed and Galois self-dual: 
\begin{theorem} \label{composition_closed_and_self_dual_implies_strong_pos}
	If $\SysSpace{A} \subseteq \Weak$ is tensor-closed and Galois self-dual, then $\SysSpace{A}=\Strong$.
\end{theorem}

The rest of the paper is devoted to proving Theorem \ref{composition_closed_and_self_dual_implies_strong_pos}.

A system is in $\Strong$ if and only if all its event matrices 
are positive semi-definite. A matrix is positive semi-definite if and only if
every principal submatrix --- a square submatrix formed by deleting a set of rows and the matching set of columns --- has non-negative determinant. Since a principal submatrix of an event matrix is also an event matrix --- of a subset of the original set of events --- this means that 
a system is not in $\Strong$ if and only if there exists an event matrix with a negative determinant.

We will need the following useful form of the determinant of a matrix:
\begin{lemma} \label{form_of_det}
	For a complex square matrix $M$ of order $m>1$,
	\begin{equation}
		m! \det M = 2 \evensum(M) - 2 \eo(M) \,,
	\end{equation}
	where
	\begin{equation}
		\evensum(M) := \sum_{\pi \in \ev{m}} \sum_{\pi' \in \ev{m}} \prod_{i=1}^m M_{\pi(i)\pi'(i)} \quad \textrm{and} \quad
		\eo(M) := \sum_{\pi \in \ev{m}} \sum_{\pi' \in \od{m}} \prod_{i=1}^m M_{\pi(i)\pi'(i)} \,,
	\end{equation}
	where $\ev{m}$ ($\od{m}$) is the set of all even (odd) permutations of $[m]:= \{1,2,\dots m\}$.
\end{lemma}
\begin{proof}
	\begin{equation}
		m! \det{M} = \sum_{a_1,a_2,\dotsc,a_m} \sum_{b_1,b_2,\dotsc,b_m} \epsilon_{a_1 a_2\dotsc a_m} \epsilon_{b_1 b_2\dotsc b_m} M_{a_1b_1} M_{a_2b_2} \dotsm M_{a_mb_m} \,.
	\end{equation}
	$\epsilon_{a_1,a_2,\dotsc,a_m}$ is only non-zero if the function $\pi:[m]\to [m]$ given by $\pi(i)=a_i$ is a permutation, and equals $+1$ if this permutation is even and $-1$ if it is odd. So we have,
	\begin{equation}
		m! \det{M} = \left( \sum_{\pi \in \ev{m}} - \sum_{\pi \in \od{m}} \right)\left( \sum_{\pi' \in \ev{m}} - \sum_{\pi' \in \od{m}} \right) \prod_{i=1}^m M_{\pi(i)\pi'(i)} \,. \label{detM_expanded}
	\end{equation}
	
	Let $s_{1m}$ be the transposition that exchanges $1$ and $m$. Note that $s_{1m} \circ \od{m} = \ev{m}$; i.e. $s_{1m}$ composed with all odd permutations is the set of all even permutations, and vice versa. Therefore,
	\begin{align}
	 & \sum_{\pi \in \od{m}} \sum_{\pi' \in \ev{m}} \prod_{i=1}^m M_{\pi(i)\pi'(i)} \nonumber
		\\* &= \sum_{\pi \in \od{m}} \sum_{\pi' \in \ev{m}} M_{\pi(1)\pi'(1)} M_{\pi(m)\pi'(m)} \prod_{i=2}^{m-1} M_{\pi(i)\pi'(i)}\nonumber
		\\ &= \sum_{\pi \in \od{m}} \sum_{\pi' \in \ev{m}} M_{\pi \circ s_{1m}(m)\pi'\circ s_{1m}(m)} M_{\pi \circ s_{1m} (1)\pi'\circ s_{1m}(1)} \prod_{i=2}^{m-1} M_{\pi \circ s_{1m}(i)\pi' \circ s_{1m}(i)}\nonumber
		\\ &= \sum_{\pi \in \ev{m}} \sum_{\pi' \in \od{m}} M_{\pi(m)\pi'(m)} M_{\pi(1)\pi'(1)} \prod_{i=2}^{m-1} M_{\pi(i)\pi'(i)}\nonumber
		\\ &= \sum_{\pi \in \ev{m}} \sum_{\pi' \in \od{m}} \prod_{i=1}^m M_{\pi(i)\pi'(i)}\nonumber
		\\ &= \eo(M) \,. \label{even_odd_is_odd_even}
	\end{align}
	Similarly,
	\begin{align*}
		\sum_{\pi \in \od{m}} \sum_{\pi' \in \od{m}} \prod_{i=1}^m M_{\pi(i)\pi'(i)} &= \sum_{\pi \in \ev{m}} \sum_{\pi' \in \ev{m}} \prod_{i=1}^m M_{\pi(i)\pi'(i)}
		\\ &= \evensum(M)\,.
	\end{align*}
	Thus, \eqref{detM_expanded} becomes
	\begin{equation}
		m! \det M = 2 \evensum(M) - 2 \eo(M) \,.
	\end{equation}
\end{proof}
\begin{lemma} \label{ee_and_eo_are_real}
	For a Hermitian matrix $M$ of order $m>1$, both $\evensum(M)$ and $\eo(M)$ are real.
\end{lemma}
\begin{proof}

	\begin{align*}
		\evensum(M)^* &= \sum_{\pi \in \ev{m}} \sum_{\pi' \in \ev{m}} \prod_{i=1}^m {M_{\pi(i)\pi'(i)}}^*
		\\ &= \sum_{\pi \in \ev{m}} \sum_{\pi' \in \ev{m}} \prod_{i=1}^m {M_{\pi'(i)\pi(i)}}
		\\ &= \evensum(M) \,.
	\end{align*}
$\det M$ is real for a Hermitian matrix and so \lemref{form_of_det} shows that $\eo(M)$ is also real.
\end{proof}
\begin{lemma} \label{change_sign_of_cos}
	Let $\theta \ne 0$ and $-\pi < \theta \le \pi$. Then there exist non-zero natural numbers $n,m \in {\mathbb{N}}$ such that
	\begin{equation}
		\cos(n \, \theta) < 0 \quad \textrm{and} \quad
		\cos(m \, \theta) \geq 0 \,.
	\end{equation}
\end{lemma}
\begin{proof}
	Since cosine is symmetric, it is sufficient to consider $\theta \in (0,\pi]$. For $\theta \in (0,\pi/2]$, we choose $m=1$ and $n$ to equal the floor
		\begin{equation}
		n  = \left\lfloor{  \frac{\pi}{2 \theta }} + 1 \right\rfloor \,.
	\end{equation}
	For $\theta \in (\pi/2,3\pi/4]$, we choose $n=1$ and $m=3$. Finally, for $\theta \in (3\pi/4,\pi]$, we choose $n=1$ and $m=2$.
\end{proof}

We now show that  any system that is in neither $\Strong$ nor $\PosEntry$ can be composed with itself a finite number of times to produce a quasi-system that is not in $\Weak$. This is the heart of the proof of Theorem 
\ref{composition_closed_and_self_dual_implies_strong_pos}. 
\begin{lemma} \label{if_not_SorP_then_not_weak_to_pow}
	If $\Sys=\System \in \Weak \setminus (\Strong \cup \PosEntry)$, then there exists $k \in {\mathbb{N}}$ such that $\Sys^{\odot k} \not \in \Weak$.
\end{lemma}
\begin{proof}
	First, we write the values of $D$ in polar form
	\begin{equation}
		D(A,B) = r(A,B)\, e^{i \theta(A,B)} \,,
	\end{equation}
	where $r$ and $\theta$ are real functions on $\EA \times \EA$,
	\begin{equation}
		r(A,B) \geq 0, \quad  \theta(A,B) \in (-\pi,\pi] \quad \textrm{and} \quad \theta(A,B)=0 \quad \text{if } r(A,B)=0\,.
	\end{equation}
	Since $D$ is Hermitian, $r$ is symmetric and $\theta$ is antisymmetric, except when $\theta(A,B)=\pi=\theta(B,A)$.  $D(E,E)=r(E,E)$ for all $E \in \EA$.
	
	Now, since $\Sys \not \in \PosEntry$, there exists some $A,B \in \EA$ such that
	\begin{equation}
		\theta(A,B) \neq 0 \,.
	\end{equation}
	We want to find two \textit{disjoint} events with the above property. 	Recalling that $(1+A)$ is the complement of $A$, we define
	\begin{equation}
		A_1 := A \cdot B\,, \quad  A_2 := A \cdot (1+B)\,, \quad B_1 := A \cdot B \quad \textrm{and} \quad B_2 := B \cdot (1+A)\,.
	\end{equation}	
	These four events are pairwise disjoint, except for $A_1 = B_1$. 
We have $A=A_1\sqcup A_2$ and $B=B_1\sqcup B_2$. Thus,
	\begin{alignat}{2}
		r(A,B)\, e^{i \theta(A,B)} &= D(A_1 \sqcup A_2,B_1\sqcup B_2)
		\\ & = D(A_1, B_1) + D(A_1, B_2) + D(A_2, B_1) + D(A_2,B_2) && \BY{bi-additivity}
		\\ & = r(A_1, A_1) + D(A_1, B_2) + D(A_2, B_1) + D(A_2,B_2) && \BY{$A_1=B_1$.}
	\end{alignat}
	Since the phase on the left-hand side is non-zero, at least one of the last three $f(\cdot,\cdot)$ terms must also have a non-zero phase. Choose one of these terms with a non-zero phase and rename the first and second arguments of that term $\bar{A}$ and $\bar{B}$ respectively.  Then
	\begin{equation}
		\bar{\theta} := \theta\Arg[\big]{\bar{A},\bar{B}} \neq 0 \,,
	\end{equation}
	with $r\Arg[\big]{\bar{A},\bar{B}} \neq 0$ and  $\bar{A}$ and $\bar{B}$ disjoint.

	In addition, since $\Sys \not \in \Strong$, there exists a finite subset $\Eset{B} \subseteq \EA$ whose event matrix $M$ is not positive semi-definite. By considering the event matrix of the set of atoms of the event algebra generated by 
	$\Eset{B}$ and using  \lemref{sum_of_strong_pos_events_is_strong_pos},  we may assume that the elements of $\Eset{B}$ are pairwise disjoint. 
There exists a principal submatrix $N$ of $M$ such that 
	\begin{equation}
		\det N < 0 \,.
	\end{equation}
Since $N$ is a principal submatrix of $M$, it is the event matrix for some $\Eset{C} \subseteq \Eset{B}$. Let 
$\Eset{C} = \{F_1, F_2, \dots F_n\}$. 
The $F_i$ are pairwise disjoint. If $n=1$ then $N = D(F_1, F_1) < 0 $ and $D$ is not weakly positive and we are done. Therefore, from now on we assume $n>1$.

\noindent Consider now $\Sys^{\odot k} = \{ \Omega^k,\EA^{\odot k}, D^{\odot k}\} $. We will find an appropriate $k$ for each of a number of cases and subcases. 

\noindent \textit{Case (a)}:  $
		r\Arg[\big]{\bar{A},\bar{A}} = r\Arg[\big]{\bar{B},\bar{B}} = 0$.
		
			Consider the event $E \in \EA^{\odot k}$ given by
	\begin{equation}
		E = {\bar{A}}^k \sqcup {\bar{B}}^k = \underbrace{\bar{A} \times \bar{A} \times \dotsm \times \bar{A}}_k 
		\sqcup \underbrace{ \bar{B} \times \bar{B} \times \dotsm \times \bar{B}}_k \,.
	\end{equation}
	$\bar{A}^k $ and $\bar{B}^k $ are disjoint. Then, 
	\begin{align*}
		D^{\odot k}(E,E) &= D^{\odot k}\Arg[\big]{\bar{A}^k,\bar{A}^k} + D^{\odot k}\Arg[\big]{\bar{B}^k,\bar{B}^k} + D^{\odot k}\Arg[\big]{\bar{A}^k,\bar{B}^k} + D^{\odot k}\Arg[\big]{\bar{B}^k,\bar{A}^k}
		\\ &= D\Arg[\big]{\bar{A},\bar{A}}^k + D\Arg[\big]{\bar{B},\bar{B}}^k + D\Arg[\big]{\bar{A},\bar{B}}^k + D\Arg[\big]{\bar{B},\bar{A}}^k
		\\ &= r\Arg[\big]{\bar{A},\bar{B}}^k e^{i k \theta\Arg{\bar{A},\bar{B}}} + r\Arg[\big]{\bar{B},\bar{A}}^k e^{i k \theta\Arg{\bar{B},\bar{A}}}
		\\ &= 2 \, r\Arg[\big]{\bar{A},\bar{B}}^k \cos\Arg[\big]{k \bar{\theta}} \,,
	\end{align*}
	where we used the symmetry of $r$ and the antisymmetry of $\theta$. Using \lemref{change_sign_of_cos} we choose $k$ such that $\cos\Arg[\big]{k \, \bar{\theta}} < 0$ so that 		$D^{\odot k}(E,E) < 0$
	and we are done.

	\noindent \textit{Case (b)}: $
		r\Arg[\big]{\bar{A},\bar{A}} + r\Arg[\big]{\bar{B},\bar{B}}>0$.
		
Let $k = p + nq$ where  $p,q$ are positive integers and consider events $E^\textup{e},E^\textup{o} \in \EA^{\odot k}$, given by
	\begin{gather}
		E^\textup{e} = \bigsqcup_{ \pi_1 \in \ev{n} }  \bigsqcup_{ \pi_2 \in \ev{n} } \dotsb  \bigsqcup_{ \pi_q \in \ev{n} } \bar{A}^{p} \times \prod_{i_1=1}^n F_{\pi_1(i_1)} \times \prod_{i_2=1}^n F_{\pi_2(i_2)} \times \dotsm \times \prod_{i_q=1}^n F_{\pi_q(i_q)}\,,
		\\ E^\textup{o} = \bigsqcup_{ \pi_1 \in \od{n} }  \bigsqcup_{ \pi_2 \in \od{n} } \dotsb  \bigsqcup_{ \pi_q \in \od{n} } \bar{B}^{p} \times \prod_{i_1=1}^n F_{\pi_1(i_1)} \times \prod_{i_2=1}^n F_{\pi_2(i_2)} \times \dotsm \times \prod_{i_q=1}^n F_{\pi_q(i_q)} \,,
	\end{gather}
	where the Cartesian products ``$\prod_{i}$'' are taken in order, from left to right, as for example in,
	\begin{equation}
	 \prod_{i_1=1}^n F_{\pi_1(i_1)}  = F_{\pi_1(1)} \times F_{\pi_1(2)} \times \dots F_{\pi_1(n)} \,.
	 \end{equation}
	The symbol $\bigsqcup$ is used because the unions are indeed over disjoint events since  $\prod_i F_{\pi(i)}$ is disjoint from $\prod_i F_{\pi'(i)}$ when $\pi$ and $\pi'$ are different permutations.

	Let $E=E^\textup{e}\sqcup E^\textup{o}$, then
	\begin{equation}
		D^{\odot k}(E,E) = D^{\odot k}(E^\textup{e},E^\textup{e}) + D^{\odot k}(E^\textup{o},E^\textup{o}) +D^{\odot k}(E^\textup{e},E^\textup{o}) +D^{\odot k}(E^\textup{o},E^\textup{e}) \BY{bi-additivity.} \label{fEE_expanded_in_even_and_odd}
	\end{equation}
	Focussing on the third term,
	\begin{align*}
		& D^{\odot k}(E^\textup{e},E^\textup{o}) \nonumber
		\\* &= \underbrace{\sum_{\pi_1 \in \ev{n}} \sum_{\pi_2 \in \ev{n}} \dotsb \sum_{\pi_q \in \ev{n}} }_{\text{from $E^\textup{e}$}} \underbrace{\sum_{\pi_1' \in \od{n}} \sum_{\pi_2' \in \od{n}} \dotsb \sum_{\pi_q' \in \od{n}}}_{\text{from $E^\textup{o}$}} \nonumber
		\\* & \quad \begin{aligned}
			D^{\odot k} \Arg[\Bigg]{ & \bar{A}^p \times \prod_{i_1=1}^n F_{\pi_1(i_1)} \times \prod_{i_2=1}^n F_{\pi_2(i_2)} \times \dotsm \times \prod_{i_q=1}^n F_{\pi_q(i_2)}\,,
			\\* & \bar{B}^p \times \prod_{i_1'=1}^n F_{\pi_1'(i_1')} \times \prod_{i_2'=1}^n F_{\pi_2'(i_2')} \times \dotsm \times \prod_{i_q'=1}^n F_{\pi_q'(i_2')}} \BY{bi-additivity}
		\end{aligned}
		\\ &= \sum_{\pi_1 \in \ev{n}} \sum_{\pi_1' \in \od{n}} \sum_{\pi_2 \in \ev{n}} \sum_{\pi_2' \in \od{n}}\dotsb \sum_{\pi_q \in \ev{n}} \sum_{\pi_q' \in \od{n}} \nonumber
		\\* & \quad D\Arg[\big]{\bar{A},\bar{B}}^p \prod_{i_1=1}^n D\Arg[\Big]{F_{\pi_1(i_1)},F_{\pi_1'(i_1)}} \prod_{i_2=1}^n D\Arg[\Big]{F_{\pi_2(i_2)},F_{\pi_2'(i_2)}} \dotsm \prod_{i_q=1}^n D\Arg[\Big]{F_{\pi_q(i_q)},F_{\pi_q'(i_q)}}
		\\ &= D\Arg[\big]{\bar{A},\bar{B}}^p \Arg[\Bigg]{ \sum_{\pi \in \ev{n}} \sum_{\pi' \in \od{n}}\prod_{i=1}^n N_{F_{\pi(i)} F_{\pi'(i)}} }^q
		\\ &= D\Arg[\big]{\bar{A},\bar{B}}^p\, \eo(N)^q \,.
	\end{align*}
	We do a similar calculation for the other terms in \eqref{fEE_expanded_in_even_and_odd}. Then, we use the result from \eqref{even_odd_is_odd_even} to change a sum over even permutations on the first index and odd permutations on the second index to a sum over odd permutations on the first index and even permutations on the second, and similarly a sum over odd permutations to a sum over even permutations. We find that
	\begin{align*}
		D^{\odot k}(E,E) &= \Arg[\Big]{ D\Arg[\big]{\bar{A},\bar{A}}^p + D\Arg[\big]{\bar{B},\bar{B}}^p } \evensum(N)^q + \Arg[\Big]{ D\Arg[\big]{\bar{A},\bar{B}}^p + D\Arg[\big]{\bar{B},\bar{A}}^p} \eo(N)^q
		\\ &= x_p\, \evensum(N)^q + y_p \cos \Arg[\big]{ p \, \bar{\theta}} \, \eo(N)^q \,,
	\end{align*}
	where
	\begin{equation}
		x_p := r\Arg[\big]{\bar{A},\bar{A}}^p + r\Arg[\big]{\bar{B},\bar{B}}^p \quad \textrm{and} \quad  y_p := 2 r\Arg[\big]{\bar{A},\bar{B}}^p \,.
	\end{equation}
	Both $x_p$ and $y_p$ are {strictly} positive real numbers.

	By \lemref{form_of_det}, we have
	\begin{equation}
		n! \det N = 2\evensum(N) - 2\eo(N) \,.
	\end{equation}
	But $\det N$ is negative and by \lemref{ee_and_eo_are_real} both $\evensum(N)$ and $\eo(N)$ are real, so
	\begin{equation}
		\evensum(N) < \eo(N) \,. \label{ee_less_than_eo}
	\end{equation}

\textit{Subcase (i)}  $\eo(N) \leq 0$.  This implies $\evensum(N) < 0$. In this case, we choose $q=1$ and use \lemref{change_sign_of_cos} to choose a $p$ such that $\cos \Arg[\big]{ p \, \bar{\theta} } \geq 0$ to get
	\begin{align*}
		D^{\odot k}(E,E) &=  { x_p \,\evensum(N) } +  y_p  \cos \Arg[\big]{ p \, \bar{\theta} } \eo(N) 
		\\ &< 0\,.
	\end{align*}

\textit{Subcase (ii)}  $\eo(N) > 0$ and $\evensum(N) \le 0$. Choose $q=1$ and $p$ such that $\cos \Arg[\big]{ p \, \bar{\theta} } < 0$ (\lemref{change_sign_of_cos}). Then 
\begin{align*}
		D^{\odot k}(E,E) &=  x_p \, \evensum(N) +  y_p  \cos \Arg[\big]{ p \, \bar{\theta} } \eo(N)
		\\ &< 0\,.
	\end{align*}

\textit{Subcase (iii)}  $\eo(N) > 0$ and $\evensum(N) > 0$. Again choose a $p$ such that $\cos \Arg[\big]{ p \, \bar{\theta} } < 0$.
	Then
	\begin{equation}
		D^{\odot k}(E,E) = x_p \eo(N)^q \Arg[\bigg]{  \Arg[\bigg]{\frac{\evensum(N)}{\eo(N)}}^q +  \frac{y_p}{x_p} \cos \Arg[\big]{ p \, \bar{\theta} }  } \,.
	\end{equation}
	Since $\frac{\evensum}{\eo} < 1$, for large enough $q$ the first term in the brackets can be made arbitrarily small, while the second term is fixed and strictly negative. So there exists  $q \in {\mathbb{N}}$ 
	for which $D^{\odot k}(E,E)$ is negative. \end{proof}

\begin{corollary}\label{sunionp}
A tensor-closed subset of $\Weak$ is a subset of  $\Strong \cup \PosEntry$. 
\end{corollary}

\begin{theorem} \label{composition_closed_subset_of_P_or_S}
	If $\SysSpace{A} \subseteq \Weak$ is tensor-closed, then either $\SysSpace{A} \subseteq \Strong$ or $\SysSpace{A} \subseteq \PosEntry$.
\end{theorem}
\begin{proof} Let $\SysSpace{A}$ be tensor-closed. 
$\SysSpace{A} \subseteq \Strong \cup \PosEntry $.
Assume, for contradiction, there exist  $\Sys[1]=\System[1]$ and $ \Sys[2]=\System[2]$ both in  $\SysSpace{A}$ such that $\Sys[1] \in \PosEntry \setminus \Strong$ and $\Sys[2] \in \Strong \setminus \PosEntry$.  $\SysSpace{A}$ is tensor-closed so $\Sys[1] \odot \Sys[2]$ is in $\SysSpace{A}$.

	Since $\Sys[1] \not \in \Strong$, there exists some finite $\Eset{B} \subseteq \EA[1]$ with a corresponding event matrix $M$ that is not positive semi-definite. Consider the finite set of events
	\begin{equation}
		\Eset{C} = \setdef{ A \times \HS[2] }{ A \in \Eset{B}} \subseteq \EA[1] \odot \EA[2] \,.
	\end{equation}
	The corresponding event matrix $N$ for $D_1 \odot D_2$ is given by
	\begin{align*}
		N_{A \times \HS[2] B \times \HS[2]} &=  D_1(A,B)\,D_2(\HS[2],\HS[2])
		\\ &= M_{AB} \,,
	\end{align*}
	which implies it is also not positive semi-definite. Therefore, $\Sys[1] \odot \Sys[2] \not \in \Strong$.

	Also, since $\Sys[2] \not \in \PosEntry$, there exists some $A,B \in \EA[2]$ such that $D_2(A,B)$ is either negative or non-real. But then
	\begin{align*}
		D_1 \odot D_2(\HS[1] \times A,\HS[1] \times B) &= D_1(\HS[1],\HS[1])\, D_2(A,B)
		\\ &= D_2(A,B) \,,
	\end{align*}
	so $D_1 \odot D_2$ also has a negative or non-real entry. Therefore, $\Sys[1] \odot \Sys[2] \not \in \PosEntry$. 
	
	But  $\Sys[1] \odot \Sys[2] \not \in \Strong$ and $\Sys[1] \odot \Sys[2] \not \in \PosEntry$ contradicts 
	Corollary \ref{sunionp}.
\end{proof}

\begin{lemma} \label{subsets_and_duals}
	For any $\SysSpace{A},\SysSpace{B} \subseteq \Weak$,
	\begin{equation}
		\SysSpace{A} \subseteq \SysSpace{B} \implies \dual{\SysSpace{B}} \subseteq \dual{\SysSpace{A}}\,.
	\end{equation}
\end{lemma}
\begin{proof}
	If $\Sys[1] \in \dual{\SysSpace{B}}$, then
	\begin{align*}
		&\Sys[1] \odot \Sys[2] \in \Weak \quad \forall \Sys[2] \in \SysSpace{B} \supseteq \SysSpace{A}
		\\ & \implies \quad \Sys[1] \odot \Sys[2] \in \Weak \quad \forall \Sys[2] \in \SysSpace{A}
		\\ & \implies \quad \Sys[1] \in \dual{\SysSpace{A}} \,.
	\end{align*}
\end{proof}

We are now ready to prove Theorem \ref{composition_closed_and_self_dual_implies_strong_pos}:

\begin{proof}
	Suppose $\SysSpace{A}$ is tensor-closed and Galois self-dual. We know from \thmref{composition_closed_subset_of_P_or_S} that either $\SysSpace{A} \subseteq \PosEntry$ or $\SysSpace{A} \subseteq \Strong$.

	If $\SysSpace{A} \subseteq \PosEntry$, then
	\begin{equation}
		{\SysSpace{A}} \subseteq {\PosEntry} \subset \dual{\PosEntry}\subseteq \dual{\SysSpace{A}} = 
		\SysSpace{A} \BY{\lemref{subsets_and_duals} and \corref{P_not_self_dual},}
	\end{equation}
	which is a contradiction.

	If $\SysSpace{A} \subseteq \Strong$, then
	\begin{equation}
		{\SysSpace{A}} \subseteq \dual{\Strong} = \Strong \subseteq \dual{\SysSpace{A}} = 
		{\SysSpace{A}} \BY{\lemref{subsets_and_duals} and \thmref{S_self_dual},}
	\end{equation}
	which implies $\SysSpace{A}=\Strong$.
\end{proof}

\section{Discussion}\label{discussion}

Our theorem adds to the evidence that strong positivity is the correct physical positivity condition on the 
decoherence functional/double path integral in both quantum measure theory and generalised quantum mechanics but it is not a proof because of the various assumptions we have made 
throughout, natural though they are.  

Why should compose-ability be a requirement at all? The physical universe is a whole and one might consider any attempt to split it up into subsystems as necessarily doing some kind of violence to it. Maybe in a truly cosmological quantum theory the question of composition of quantum systems might not arise but for now
it is hard to see how we can make progress without considering subsystems, both in isolation from and in interaction with others. One can consider the concept of a set of physically allowed systems, closed under composition as some sort of combined locality-cum-reproducibility requirement of a physical theory and it is all but universally assumed. As an example, the condition of compose-ability discriminates between different decoherence 
conditions in generalised quantum mechanics and in the decoherent histories approach 
to quantum foundations 
 \cite{PhysRevLett.92.170401, Halliwell:2009}.

It is also worth bearing in mind the possibility that the product composition law that seems natural for pure, unentangled states may not be, or may not always be, the appropriate composition law. As an illustration of the subtle issues that can arise when generalising from classical stochastic measure theories to quantum measure theories, consider the fact that even the product composition law for decoherence functionals of pure, unentangled states can result in 
correlations between events in the subsystems if one adopts the preclusion law that events of 
zero measure do not happen  \cite{Wallden:2012nh}. This is what Sorkin refers to as the ``radical inseparability'' 
of quantum systems analysed from the perspective of the path integral \cite{Sorkin:2013kxa}.

The issue of composition is 
intertwined, in quantum measure theory, with the question of the relationship between the 
complex decoherence functional and the real, non-negative quantum measure. Whilst the imaginary part of the decoherence functional of a single system does not 
affect the quantum measure of that system, if one composes the decoherence functionals of two 
subsystems using the product composition rule then the quantum measure of the composed system will depend on the imaginary parts of the decoherence functionals of the subsystems. This means that we 
cannot compose quantum measure systems by composing their quantum measures directly but must do it by composing their decoherence functionals. 

The previous remarks notwithstanding, one can nevertheless conceive of a theoretical landscape of quantum measure theories given to us, somehow, only by their measures and not by their decoherence functionals. How would one compose systems in this case? Sorkin showed \cite{Sorkin:1994dt} that there is a one-to-one correspondence between quantum measures, $\mu$, and real symmetric decoherence functionals, $D^\mu$ given by: 
 	\begin{align*}
\mu(\cdot) \mapsto &D^\mu(\cdot, \cdot) \,, \\
\textrm{where}, \ &D^\mu(A, B):= \frac{1}{2} \Big[\mu(A\cup B) + \mu(A\cap B) - \mu(A \setminus B) - \mu(B \setminus A)\Big]\ \ \forall A, B\in \EA \,.
	\end{align*}
To define the composition of two quantum measures, $\mu_1$ and $\mu_2$, then, one can form their real symmetric decoherence functionals, $D^{\mu_1}$ and $D^{\mu_2}$, compose these decoherence functionals and finally form the quantum measure from this composition. 
Now, in this landscape, all decoherence functions are real, and we can redo the work in this paper replacing the set $\Weak$ with $\Weak_R$ which is the subset of $\Weak$ with real decoherence functionals. Almost all the lemmas and theorems --- mutatis mutandis --- still hold, including  theorem \ref{composition_closed_subset_of_P_or_S}. The only thing that fails is the 
final result because the replacement of
$\Weak$ with $\Weak_R$ in the definition of Galois dual has the effect of making the 
set $\PosEntry$ of positive entry systems Galois self dual as well as tensor closed. So in a landscape of systems  with real decoherence functionals, our uniqueness theorem for strong positivity, theorem \ref{composition_closed_and_self_dual_implies_strong_pos} fails. 

One motivation for this work was that it might have an application or extension at the higher, super-quantum levels of the Sorkin nested hierarchy of  measure theories  \cite{Sorkin:1994dt}. We have shown that the set of strongly positive systems $\Strong$ is the unique set of quantum systems that is tensor-closed and Galois self-dual. This may be a useful clue for finding the correct, physical positivity condition for measure theories at levels of the Sorkin hierarchy above the quantum level. Strong positivity is a condition on the decoherence functional, $D$,  and not (directly) on the measure, $\mu$ and as we have seen above, we need decoherence  functionals to compose systems. So, to investigate composition of systems and the analogue of the strong positivity condition at higher levels, we need the analogue of the decoherence functional at higher levels.  

Consider for example level 3, the first super-quantum level. Given a complex, level 3 decoherence functional with \textit{three} arguments, $E(A, B, C)$ such that the functional is additive in each of its arguments, 
\begin{align*}
E(A_1\sqcup A_2, B, C) = E(A_1, B, C) + E(A_2, B, C) \,\quad \forall A_1,A_2, B, C \in \EA \,
\end{align*} 
and similarly for the other two arguments,  and such that 
\begin{align*}
E(A, A, A) \ge 0 \,\quad \forall A \in \EA \,,
\end{align*} 
then the measure, $\mu(A): = E(A,A,A)$ satisfies the level three Sorkin sum rule  \cite{Sorkin:1994dt}. 
There, however, the easy generalisations from the quantum level in the hierarchy end and a number of questions arise. 

The Sorkin hierarchy is nested so each level is contained in all higher levels. Thus, a classical, level 1 theory is a special case of a quantum, level 2 theory in this measure theoretic framework for 
classifying physical theories. Indeed, this is one of the reasons for expecting that a path integral framework is the right one for understanding how classical  physics emerges approximately from a fundamentally quantum theory. This nested relationship between classical and quantum measure theories can be expressed in terms of the decoherence functional in the following way.  If a quantum/level 2  decoherence functional satisfies 
$D(A,B) = D(A\cap B, A\cap B)$ for all events $A$ and $B$ in the event algebra then the measure 
defined by $\mu(A): = D(A,A)$  is classical \textit{i.e.} it satisfies the level 1, Kolmogorov sum rule. Conversely, given a 
classical measure $\mu$, one can define a decoherence functional: $D(A,B): = \mu(A\cap B)$. 
Now consider level 3, the first super-quantum level. Any level 2 measure is a special case of a level 3 measure: it satisfies the level 3 sum rule (and all higher level sum rules). But, can this inclusion of level 2 in level 3 be expressed in terms of decoherence functionals? Given a level
2 decoherence functional, $D(A,B)$,  can a level 3 functional, $E(A,B,C)$, be defined using $D$, such that $E(A,A,A) = D(A,A) $, \textit{i.e.} $E$ corresponds to the same measure?
What condition should replace the quantum/level 2 condition of Hermiticity? How do we describe the composition of two level 3 systems: is the same product rule as 
employed in this paper the right rule? What is the correct physical positivity condition on a level 3 decoherence functional? One strategy for discovering this condition, suggested by the results of this paper, is to seek a set of level 3 systems that is closed under composition and is maximal amongst such sets, in the hope that this may again prove to be a unique set, whose elements are characterised by a property one can recognise as a positivity condition. 

Finally, let us address Roger's particular concerns in quantum foundations in the context of the difference between generalised quantum mechanics (GQM) and quantum measure theory (QMT). In 1994 in a debate with Stephen Hawking,  Roger said \cite{HawkingPenrose1995}: 
\begin{quote}
Whatever ``reality'' may be, one has to explain how one perceives the world to be. [...] It seems to me that in order to explain how we perceive the world to be, within the framework of Quantum Mechanics, we shall need to have one (or both) of the following:\\
\begin{itemize}
\item[(A)] A theory of experience.
\item[(B)] A theory of real physical behaviour. 
\end{itemize}
\end{quote} 
 Roger goes on to state that he is  ``a strong B-supporter''. 

How do the two path integral approaches to quantum foundations QMT and GQM  fare when judged against  Roger's  (A) and/or (B)?   The main distinction between QMT and GQM is a fork in the road signposted by an attitude to measure theory in physics.  In GQM, the attitude taken is that a classical (level 1) measure is necessary to do physics and so the full event algebra must be restricted to a subalgebra on which the measure is classical.  In contrast, in QMT, the attitude taken is that the null---and very close to null---events exhaust the scientific content of a measure theory via what Borel called ``la loi unique du hasard'' (the only law of chance) namely that events of very small measure almost certainly don't happen.\footnote{See sections 2.2 and 6.2.2  of \cite{Shafer_2006} for a historical perspective on this view---also known as Cournot's Principle---in the context of an assessment of Kolmogorov's contributions to the foundations of probability theory.} In which case, additivity of the measure is not necessary. 

In GQM, then, one seeks a maximal subalgebra of the full algebra of events such that the measure restricted to that subalgebra is a classical measure to a very good degree of approximation. The set of atoms of such a subalgebra is a maximally fine-grained, decoherent set of coarse-grained histories, in the terminology of GQM \cite{Hartle:1992as, Hartle:1998yg, Hartle:2006nx}.  One then interprets the measure restricted to a decoherent subalgebra as a probability in the usual way as for a classical random process: exactly one atom of that subalgebra is realised, at random and the  measure of any event in the decoherent subalgebra is interpreted as the probability that that event happens, i.e. the probability that the single realised atom is a subevent of that event. In the case when the decoherence functional corresponds to an initial Schroedinger cat-type state for some macro-object, a pointer say, then -- heuristics and model calculations show -- there will be a decoherent subalgebra that contains amongst its atoms, one atom in which the pointer is in one of the positions of the superposition and another atom in which the pointer is in the other position. From this decoherent subalgebra, only one atom is realised -- either one or the other of the pointer positions -- which seems to indicate that GQM is a Penrose B-type theory. 

However, in GQM, for any system there are many ---  infinitely many ---  incomparable decoherent subalgebras, all on the same footing according to the axioms of GQM. If one atom is realised from the pointer subalgebra  then one atom is realised from  \textit{each} of the decoherent subalgebras. 
\cite{Dowker:1994dd,PhysRevLett.75.3038}. Without an extra axiom, a criterion for selecting one of the decoherent subalgebras from the many, 
 GQM is therefore a theory of many worlds.\footnote{These are not the same many worlds as in the Everett interpretation.} Those who claim that GQM is satisfactory in the absence of a physical subalgebra selection axiom must construct arguments  to try to explain why we nevertheless experience only one world. There is no consensus on whether the arguments that exist in the literature hold water but, it seems to me, there is consensus that such arguments are needed. GQM is therefore an A-type theory in the Penrose sense of needing to be supplemented by a theory of experience. 
 
By contrast, QMT is a One World theory  in which the physical world is conjectured to be exactly one \textit{co-event} or generalised history \cite{Sorkin:2006wq, Sorkin:2007uc, Sorkin:2010kg} in which every event in the full event algebra for the system either happens (is affirmed) or doesn't happen (is denied). The term co-event reflects that this physical information can be considered, mathematically, as a map from the event algebra to $\mathbb{Z}_2 = \{0, 1\}$ where  $1$ represents affirmation and $0$ represents denial. The theory provides the set of physically allowed co-events and exactly one of these corresponds to the physical world. The quantum measure restricts the possible physical co-events by the Law of Preclusion that null events are denied: $\mu(E)=0 \implies \phi(E) = 0$. 
This Law of Preclusion must be supplemented by other axioms for physically allowed co-events and the question of what these axioms are remains open, though proposals have been made and explored (see for example \cite{Sorkin:2006wq, Sorkin:2007uc, Sorkin:2010kg, Dowker:2007kz, Sorkin:2013kxa, Dowker_2017, wilkes2018evolving}). The ongoing search for a physical co-event scheme is guided by  several desiderata, including the requirement that the physically allowed co-events turn out to be classical when restricted to the subalgebra of localised, quasi-classical, macro-events.\footnote{A co-event is classical if and only if it is a homomorphism from its domain to $\mathbb{Z}_2$.} This would imply that exactly one atom of the macro-subalgebra is affirmed. QMT is a One World theory, which world should recover a classical picture when restricted to the sub-algebra of macro-events.  QMT is a Penrose B-type theory.

\section{Acknowledgments}
The results proved here are proved in the PhD thesis of HW \cite{wilkes:2019}. We thank Rafael Sorkin for introducing us to the concept of Galois connection and for helpful discussions. This research was supported in part by Perimeter Institute for Theoretical Physics. Research at Perimeter Institute is supported
by the Government of Canada through Industry Canada and by the
Province of Ontario through the Ministry of Economic Development
and Innovation.  
FD is supported in part by STFC grants ST/P000762/1 and ST/T000791/1 and by APEX grant APX${\backslash}$R1${\backslash}$180098. 
HW was supported by STFC grant ST/N504336/1.

\bibliography{../BIBLIOGRAPHY/refs}
\bibliographystyle{../BIBLIOGRAPHY/JHEP}

\end{document}